\documentclass{article}

\usepackage{arxiv}

\usepackage[utf8]{inputenc} 
\usepackage[T1]{fontenc}    
\usepackage{hyperref}       
\usepackage{url}            
\usepackage{booktabs}       
\usepackage{amsfonts}       
\usepackage{nicefrac}       
\usepackage{microtype}      
\usepackage{lipsum}		
\usepackage{graphicx}
\usepackage{natbib}
\usepackage{doi}

\usepackage{amssymb}
\usepackage{xcolor}
\usepackage{amsmath}
\usepackage{amsthm}
\usepackage{subcaption}

\newtheorem{assumption}{Assumption}
\newtheorem{theorem}{Theorem}
\newtheorem{definition}{Definition}

\newtheorem{proposition}{Proposition}
\newtheorem{prob}{Problem}
\newtheorem{rem}{Remark}

\title{A Smooth Penalty-Based Feedback Law for Reactive Obstacle Avoidance with Convergence Guarantees}

\date{} 					

\author{ Lyes Smaili\\
	Department of Computer Science and Engineering\\
	University of Quebec in Outaouais\\
	101 St-Jean Bosco, Gatineau, QC, J8X 3X7, Canada \\
	\texttt{smal01@uqo.ca}\\
	\And
	 Soulaimane Berkane\thanks{Soulaimane Berkane is also with the Department of Electrical Engineering, Lakehead University, Thunder Bay, ON P7B 5E1, Canada.}\\
	Department of Computer Science and Engineering\\
	University of Quebec in Outaouais\\
	101 St-Jean Bosco, Gatineau, QC, J8X 3X7, Canada \\
	\texttt{soulaimane.berkane@uqo.ca}\\
}
\funding{This research work is supported in part by NSERC-DG RGPIN-2020-04759 and Fonds de recherche du Qu\'ebec (FRQ).}

\renewcommand{\shorttitle}{}

\hypersetup{
pdftitle={A Smooth Penalty-Based Feedback Law for Reactive Obstacle Avoidance with Convergence Guarantee},
pdfauthor={Lyes Smaili, Berkane Soulaimane},
pdfkeywords={Reactive navigation, Obstacle avoidance, Penalty-based control, Smooth feedback, Safety-critical systems. },
}

\begin{document}
\maketitle

\begin{abstract}
	This paper addresses the problem of safe autonomous navigation in unknown obstacle-filled environments using only local sensory information. We propose a smooth feedback controller derived from an unconstrained penalty-based formulation that guarantees safety by construction. The controller modifies an arbitrary nominal input through a closed-form expression. The resulting closed-form feedback has a projection structure that interpolates between the nominal control and its orthogonal projection onto the obstacle boundary, ensuring forward invariance of a user-defined safety margin. The control law depends only on the distance and bearing to obstacles and requires no map, switching, or set construction. When the nominal input is a gradient descent of a navigation potential, we prove that the closed-loop system achieves almost global asymptotic stability (AGAS) to the goal. Undesired equilibria are shown to be unstable under a mild geometric curvature condition, which compares the normal curvature of the obstacle boundary with that of the potential level sets. We refer to the proposed method as SPF (Safe Penalty-based Feedback), which ensures safe and smooth navigation with minimal computational overhead, as demonstrated through simulations in complex 2D and 3D environments.
\end{abstract}


\section{Introduction}
\subsection{Motivation}
Autonomous navigation in cluttered and unknown environments remains a fundamental challenge in robotics, with critical applications in ground, aerial, and underwater platforms. A key objective is to avoid collisions in real time using only local sensor information, without relying on prior maps or global planning. Although significant progress has been made in designing reactive navigation strategies, existing approaches often fall short of simultaneously ensuring simplicity, smoothness, and global guarantees. Many rely on complex logic, discontinuities, or case-based switching that hinder their integration into real-time systems with higher-order dynamics. This motivates the development of a purely reactive, closed-form controller that is both smooth and easy to implement, while still ensuring safety and global convergence to the goal. Our work proposes such a solution, using an unconstrained penalty-based feedback that blends seamlessly with the nominal control law and adapts continuously to obstacle proximity.

\subsection{Related Work}
Various methods have been developed to address the challenges of autonomous navigation. A broad class of these methods assumes prior knowledge of the environment and falls under the umbrella of \emph{global} or \emph{map-based} approaches. These include discrete path-finding methods based on grid or graph representations of the environment (e.g., Dijkstra, A*), as well as sampling-based algorithms such as Rapidly Exploring Random Trees (RRTs) and Probabilistic Roadmaps (PRMs). Feedback motion planning techniques also belong to this class \citep{lavalle2006planning}.

To overcome local minima issues in classical artificial potential fields \citep{khatib1986,Koditschek1987}, the \emph{navigation function} framework was introduced for topologically simple environments such as Euclidean sphere worlds \citep{KoditschekRimon1992}. This approach has been extended to multi-agent systems \citep{Tanner2005FormationSO,Dimarogonas2006agents}. While navigation functions typically require careful parameter tuning to eliminate local minima, newer techniques such as \emph{navigation transforms} \citep{Loizou2017} and \emph{prescribed performance control} \citep{Vrohidis2018} address this issue without the need for such tuning. Hybrid feedback approaches have also been proposed to eliminate undesired equilibria and ensure global asymptotic stability, as explored in \citep{sanfelice2006robust,berkane2019hybrid,casau2019hybrid,berkane2021obstacle}. However, a common limitation of all these approaches is their reliance on complete knowledge of the environment to guarantee both safety and convergence.

Reactive methods have been developed to enable navigation in \emph{unknown} environments, with a stronger emphasis on the autonomy and adaptability of robotic systems. Early strategies such as the \emph{Bug algorithms} \citep{Bug1986,Bug2005} offer intuitive solutions but are difficult to generalize to higher-dimensional environments. For topologically simple obstacle configurations, navigation function-based approaches have been extended to unknown settings in \citep{Lionis2007,Filippidis2011}. The sensor-based method in \citep{arslan2019sensor} uses separating hyperplanes to identify the locally free space in which the robot can navigate safely, providing almost-global asymptotic stability (AGAS) guarantees under strong convexity assumptions. Hybrid feedback controllers for unknown planar environments have also been explored in \citep{sawant2023hybrid,Sawant2023,cheniouni2024hybrid,sawant2024n}, achieving safe navigation around convex and non-convex obstacles. While effective in their respective settings, these approaches are typically only continuous and rely on switching logic or complex geometric constructions from sensor data.

\emph{Control barrier functions} (CBFs) provide another reactive framework for enforcing safety in real time, particularly in safety-critical systems \citep{WIELAND2007462}. These methods encode safety constraints that are enforced via online optimization, typically through quadratic programs (QPs). The combination of CBFs with control Lyapunov functions (CLFs) in a unified QP-based formulation has been explored in \citep{Ames2017}. However, it has been shown in \citep{Reis2021} that QP-based controllers with CBF–CLF constraints may introduce undesired asymptotically stable equilibria. Extensions to address such limitations include optimization-based reactive control via quasi-conformal mappings \citep{Notomista2025} and robust safety-critical control using reduced-order models for high-dimensional systems \citep{Molnar2023}. Despite these advances, the complexity of the required geometric mappings and optimization formulations suggests that further work is needed toward simple and scalable feedback designs.

\subsection{Contributions and Organization of the Paper}
Building on the limitations of existing methods for reactive navigation, we propose a unified framework that guarantees both safety and convergence using only local information. This work is motivated by our earlier development of the Safety Velocity Cone (SVC) method \citep{berkane2021Navigation,Smaili2024}, which uses set invariance principles \citep{Nagumo} to achieve safe navigation in arbitrary-dimensional environments based solely on local sensing. While effective in ensuring safety, the resulting controller in \citep{Smaili2024} is only Lipschitz-continuous, and thus less compatible with systems requiring differentiable feedback. In this paper, we extend this approach by formulating a smooth controller as the solution to an unconstrained penalty-based optimization problem, yielding a unified and tunable feedback structure.

The contributions of this work are fourfold. First, we derive an arbitrarily smooth feedback law, referred to as SPF (Safe Penalty-based Feedback), from an unconstrained penalty-based formulation. In contrast to hybrid and set-based approaches, which typically yield only at best continuous controllers \citep{arslan2019sensor,sawant2023hybrid,Cheniouni2025}, the proposed design ensures smoothness, making it well suited for integration with systems involving higher-order dynamics or trajectory tracking. Note that, in contrast to classical penalty-based methods \citep{mestres2022optimization}, which rely on a fixed penalty parameter, our formulation introduces a state-dependent penalty scaling function that adapts continuously to the robot's position and orientation relative to nearby obstacles. This structure enhances responsiveness near boundaries while preserving nominal performance in free space.
Second, the proposed feedback depends only on two quantities: the distance to the closest obstacle and the bearing relative to its boundary, and avoids the need for prior maps, switching logic, or geometric preprocessing. This results in a simple and lightweight solution amenable to real-time implementation. Third, safety is guaranteed for any nominal input, and when the nominal controller is a gradient descent of a potential function, we prove almost global asymptotic stability (AGAS) of the closed-loop system. The resulting law acts as a minimally invasive modification that preserves the goal-seeking behavior of the original policy. Fourth, we establish instability of all undesired equilibria under a relaxed geometric curvature condition. Unlike earlier strong-curvature assumptions \citep{arslan2019sensor,koditschek2017,Cheniouni2025}, our condition requires only the existence of a single tangent direction satisfying the inequality, making it significantly less conservative—especially in high-dimensional environments.  We validate the performance of the proposed SPF approach through simulations in complex 2D and 3D environments.

The structure of the paper reflects the main components of the proposed framework. Section~\ref{section:Notation} introduces the notation and geometric preliminaries. Section~\ref{section:Problem formulation} defines the navigation problem in unknown environments and formulates the feasibility conditions based on minimal local sensing. Section~\ref{section:main results} presents the main theoretical contributions: a smooth penalty-based feedback law that guarantees safety for arbitrary nominal inputs, preserves goal-seeking behavior, and ensures almost global asymptotic stability under a simple geometric condition—namely, that the obstacle boundary is locally more curved than the level sets of the potential function. Section~\ref{section:Simulation results} provides numerical validation in complex 2D and 3D environments. Section~\ref{section: Conclusion} concludes with a summary and future research directions.

\section{Notation}\label{section:Notation}
Let $\mathbb{R}$, $\mathbb R_{>0}$, and $\mathbb N$ denote the set of reals, positive reals, and natural numbers, respectively. Let $\mathbb R^n$ be the $n$-dimensional Euclidean space, with $n\in\mathbb N$. Let $\|.\|$ denote the Euclidean norm operator. For a subset $\mathcal{A}\subset\mathbb R^n$, let $\textbf{int}(\mathcal{A})$, $\partial\mathcal{A}$, $\overline{\mathcal{A}}$ and $\complement\mathcal{A}$ denote the topological interior, boundary, closure and complement of $\mathcal{A}$ in $\mathbb R^n$, respectively. For a function $f(x):\mathbb{R}^n\to\mathbb{R}$, let $\nabla f(x)$ and $\nabla^2f(x)$ denote the gradient and Hessian of $f$ with respect to $x$, respectively. For a vector field $F(x):\mathbb{R}^n\to\mathbb{R}^m$, let $\nabla F(x)$ represent the Jacobian of $F$ with respect to $x$. Define the distance from a point $x\in\mathbb R^n$ to a closed set $\mathcal{A}\subset\mathbb R^n$ as
$d_\mathcal{A}(x):=\inf _{y\in\mathcal{A}}\|y - x\|$. The projection of a point $x\in\mathbb{R}^n$ onto a set $\mathcal{A}\subset\mathbb R^n$ is denoted by 
$\mathbf{P}_\mathcal{A}(x):=\{y\in\overline{\mathcal{A}}: \|y-x\|=d_\mathcal{A}(x)\}$. The Euclidean ball of radius $r>0$ centered at $x$ is defined as $\mathcal{B}(x,r):=\{y\in\mathbb R^n: \|x-y\|<r\}$. Let $\mathcal{A}\subseteq\mathbb{R}^n$ be an open set. A function $f(x)$ defined on the set $\overline{\mathcal{A}}$ is of class $\mathcal{C}^k$, if it is $k$-times continuously differentiable on $\mathcal{A}$, \textit{i.e.,} all partial derivatives of $f$ up to order $k$ exist and are continuous in $\mathcal{A}$.

\section{Problem formulation}\label{section:Problem formulation}
We consider a ball-shaped robot of radius \( R > 0 \), whose center is located at \( x \in \mathbb{R}^n \). The robot operates in a workspace denoted by \( \mathcal{W} \subset \mathbb{R}^n \), which is assumed to be a subset of the Euclidean space. Let \( \mathcal{O}_i \subset \mathcal{W} \), for \( i \in \{1, \ldots, m\} \), denote \( m \) closed subsets representing the obstacles present in the workspace. Let $\mathcal{X}$ be the {\it free space} given by
\begin{equation}
    \mathcal{X}:=\mathcal{W}\setminus\bigcup_{i=1}^{m}\mathbf{int}(\mathcal{O}_i).
\end{equation}
The obstacle region is given by the complement of the free space $\complement\mathcal{X}$. Consider the distance function \( d_{\complement\mathcal{X}}(x) \), which measures the Euclidean distance from a point \( x \in \mathbb{R}^n \) to the obstacle set \( \complement\mathcal{X} \). The following assumption imposes a regularity condition on this function in a neighborhood of the boundary of the free space.

\begin{assumption}[Smoothness of the distance]\label{assumption:smoothBoundaries}
Given the free space \( \mathcal{X} \), there exists a constant \( \rho > R \) such that the distance function \( d_{\complement\mathcal{X}} \) is of class \( \mathcal{C}^k \) on the set
$\{ x \in \mathcal{X} \mid R \leq d_{\complement\mathcal{X}}(x) < \rho \}.$
\end{assumption}
This assumption does not require the boundary \( \partial\mathcal{X} \) to be smooth. Instead, it postulates the existence of a tubular region surrounding the boundary in which the distance function is sufficiently regular. The lower bound $R$ accounts for the robot’s physical size, ensuring that the center of the robot remains at least a distance $R$ away from obstacles. The parameter \( k \) specifies the degree of smoothness needed to define differential quantities such as the gradient and Hessian of \( d_{\complement\mathcal{X}} \), which are central to the design and analysis of the smooth feedback laws developed in this work. The local regularity of the distance function depends on geometric properties of the obstacle set, including the class of the set $\mathcal{X}$. While a full discussion of these geometric conditions is beyond the scope of this paper, we refer the reader to \cite{delfour2011shapes} for a comprehensive treatment of the regularity properties of distance functions and related projection operators.

When the projection of a point \( x \in \mathcal{X} \) onto the boundary \( \partial\mathcal{X} \) is unique, the distance function is differentiable at \( x \), and its gradient is given by
\begin{equation}\label{eq:gradientOfdistance}
    \nabla d_{\complement\mathcal{X}}(x) =
    \frac{x - \mathbf{P}_{\partial\mathcal{X}}(x)}{\|x - \mathbf{P}_{\partial\mathcal{X}}(x)\|}, \quad x \in \mathbf{int}(\mathcal{X}).
\end{equation}
This expression defines a unit vector pointing outward from the obstacle set. To formalize the region where this projection is well defined, we introduce the following assumption:
\begin{assumption}[Uniqueness of projection]\label{assumption:uniqeProjection}
There exists a constant \( h > R \) such that for all \( x \in \mathcal{X} \) satisfying \( d_{\complement\mathcal{X}}(x) < h \), the projection \( \mathbf{P}_{\partial\mathcal{X}}(x) \) is unique.
\end{assumption}
This condition guarantees that, in a region surrounding the obstacles and wide enough to accommodate the robot’s body, the projection map is single-valued and continuous, ensuring that the gradient of the distance function is well defined. See Fig.~\ref{fig:notationsAndAssumptions} for a schematic illustration.
\begin{figure}[h!]
    \centering
    \includegraphics[width=0.5\linewidth]{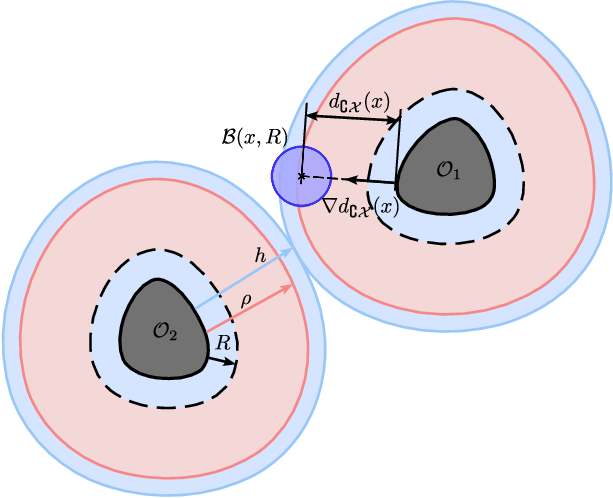}
\caption{Illustration of Assumptions \ref{assumption:smoothBoundaries}--\ref{assumption:uniqeProjection}. The (gray) shapes represent the obstacles, while the (blue) disk denotes the robot, modeled as a ball of radius \( R \) centered at \( x \). The (red) layer around the obstacles conceptually indicates the region where the distance function \( d_{\complement\mathcal{X}} \) is of class \( \mathcal{C}^k \), as required in Assumption~\ref{assumption:smoothBoundaries}. The (light blue) region shows where the projection \( \mathbf{P}_{\partial\mathcal{X}}(x) \) is guaranteed to be unique, as stated in Assumption~\ref{assumption:uniqeProjection}. The (dashed) annular layer illustrates that the robot's radius \( R \) is small enough to ensure that its center remains in the regions satisfying the assumptions.}

    \label{fig:notationsAndAssumptions}
\end{figure}

Let $\epsilon\in\mathbb{R}_{>0}$ be a design parameter that represents a safety margin from the obstacle set. We define the \textit{practical free space} $\mathcal{X}_\epsilon$ as the erosion of the environment $\mathcal{X}$ by a ball of radius $(R+\epsilon)$, that is

\begin{equation}\label{eq:practicalfreespace}
    \mathcal{X}_{\epsilon}:=\{x\in\mathbb{R}^n: d_{\complement\mathcal{X}}(x)\ge R+\epsilon\}\subset\mathcal{X}.
\end{equation}
To ensure that the smoothness and projection properties postulated in Assumptions~\ref{assumption:smoothBoundaries}--\ref{assumption:uniqeProjection} hold in a neighborhood of the boundary of $\mathcal{X}_\epsilon$, we impose the following feasibility condition on $\epsilon$:
\begin{equation}\label{condition:1}
    0<\epsilon<\min(h,\rho)-R.
\end{equation}
This condition can always be satisfied by selecting a sufficiently small \( \epsilon > 0 \), thanks to Assumptions~\ref{assumption:smoothBoundaries}--\ref{assumption:uniqeProjection}, which ensure that the quantities \( \rho \) and \( h \) are strictly greater than \( R \). To streamline notation, we define the following shorthand:
\begin{align}
    d(x) &:= d_{\complement\mathcal{X}}(x) - (R + \epsilon),\\
    \eta(x) &:= \nabla d_{\complement\mathcal{X}}(x),\\
    \mathbf{H}_d(x) &:= \nabla^2 d_{\complement\mathcal{X}}(x).
\end{align}
Finally, the robot is modeled as a first-order system evolving in the practical free space \( \mathcal{X}_\epsilon \), with dynamics
\begin{equation}\label{eq:dynamicalSystem}
    \dot{x} = u,
\end{equation}
where \( x \in \mathbb{R}^n \) denotes the robot’s center and \( u \in \mathbb{R}^n \) is the velocity control input. In our formulation, the only information about the environment required for obstacle avoidance is the distance \( d(x) \) and the unit normal vector \( \eta(x) \), both of which can be obtained from onboard sensors such as LiDAR, depth cameras, or stereo vision, using standard geometric processing techniques. The problem is then stated as follows.
\begin{prob}\label{prob:1}
Given the system \eqref{eq:dynamicalSystem}, and under Assumptions~\ref{assumption:smoothBoundaries}--\ref{assumption:uniqeProjection}, design a smooth feedback controller \( u = \kappa(x, x_d, d(x), \eta(x)) \), such that the closed-loop system
\begin{equation}\label{eq:closedLoopsystem}
    \dot{x} = \kappa(x, x_d, d(x), \eta(x))
\end{equation}
ensures that the practical free space \( \mathcal{X}_\epsilon \) is forward invariant and that the desired target position \( x_d \in \operatorname{int}(\mathcal{X}_\epsilon) \) is \emph{almost globally asymptotically stable (AGAS)}.
\end{prob}

The control policy \( \kappa \) must rely only on the robot’s current position \( x \), the target \( x_d \), and the local geometric quantities \( d(x) \) and \( \eta(x) \). For simplicity, we henceforth write \( \kappa(x) \) in place of the full expression when the dependence on \( x_d \), \( d(x) \), and \( \eta(x) \) is clear from context.

\section{Main Results}\label{section:main results}
Our goal is to ensure safe robot navigation within the practical free space $\mathcal{X}_\epsilon$ by guaranteeing forward invariance of this set and convergence to the desired target position $x_d$. Additionally, we seek to design a controller that deviates minimally from the nominal control input $\kappa_0(x)$-a property known as \emph{minimal invasiveness}.

To achieve this, we formulate a control strategy based on the solution of the following unconstrained optimization problem:
\begin{equation}\label{eq:optimization}
\min_u \frac{1}{2}||u-\kappa_0(x)||^2 + \frac{1}{2}\psi(d(x),s(x))(u^\top\eta(x))^2,
\end{equation}
where $\psi:\mathbb{R}\times\mathbb{R}\to\mathbb{R}_{\ge 0}$ is a $\mathcal{C}^l$-class function and 
\begin{equation}
    s(x):=\kappa_0(x)^\top\eta(x)
\end{equation}
denotes the component of the nominal control in the direction of the obstacle gradient. The cost function consists of two terms: the first term $\frac{1}{2}\|u-\kappa_0(x)\|^2$ aims to minimize deviation from the nominal control $\kappa_0(x)$, while the second term introduces a directional penalty that discourages unsafe motion near obstacles by penalizing the projection of the control $u$ onto the obstacle gradient direction $\eta(x)$.  The scalar scaling function $\psi(\cdot)$ modulates the intensity of this penalty based on the robot's proximity to obstacles and the direction of the nominal control. The properties of this function are defined below:
\begin{definition}[$\mathcal{C}^l$-Penalty Scaling Function]\label{def:Penalty function}
Let $\mu > 0$ and $\nu > 0$ be two positive real parameters. A function $\psi : \mathbb{R} \times \mathbb{R} \to \mathbb{R}_{\ge 0}$ is called a \emph{$\mathcal{C}^l$-penalty scaling function} (with parameters $\mu$ and $\nu$)\footnote{For brevity, the dependence of $\psi$ on $\mu$ and $\nu$ is omitted in the notation.} if it satisfies the following properties:
\begin{enumerate}
    \item $\psi$ is of class $\mathcal{C}^l$ on $\mathbb{R} \times \mathbb{R}$;
    \item $\psi(d,s) = 0$ whenever $d \ge \mu$ or $s \ge \nu$;
    \item $\psi(d,s) \to +\infty$ as $d \to 0^{-}$ and $s \to 0^{-}$ simultaneously.
\end{enumerate}
\end{definition}

The conditions imposed on the penalty scaling function $\psi(d(x), s(x))$ ensure that the penalty vanishes whenever the robot is sufficiently far from the obstacles or when the nominal controller points away from them, \textit{i.e.}, when either $d(x) \ge \mu$ or $s(x) \ge \nu$. In this case, the solution of the optimization problem~\eqref{eq:optimization} remains close to the nominal control.

On the other hand, if the robot approaches an obstacle while the nominal controller points toward it, the penalty scaling function $\psi(d(x), s(x))$ grows rapidly. The term $(u^\top \eta(x))^2$ then acts to counterbalance this growth by driving the control direction $u$ to become increasingly orthogonal to $\eta(x)$, reducing the projection. In the limiting case where $\psi(d(x), s(x)) \to \infty$, the optimal control must satisfy $u^\top \eta(x) = 0$, ensuring that the system avoids further motion toward the obstacle. Intuitively, $\psi$ behaves like a \textit{directional shield}, activating only when the robot is facing the obstacle, rather than surrounding it entirely. This anisotropic behavior is illustrated in Figure~\ref{fig:psi}, where the penalty structure resembles a shield or helmet, deployed selectively in front of the obstacle to prevent unsafe penetration.

\begin{figure}[h!]
    \centering
    \begin{subfigure}[b]{0.48\linewidth}
        \centering
        \includegraphics[width=\textwidth]{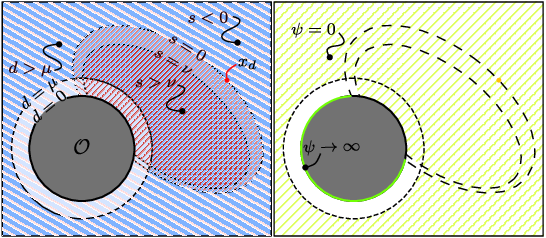}
        \caption{}
        \label{fig:2D psi vizualization}
    \end{subfigure}
    \hfill
    \begin{subfigure}[b]{0.48\linewidth}
        \centering
        \includegraphics[trim={1cm 0.25cm 1.3cm 1.5cm},clip,width=\textwidth]{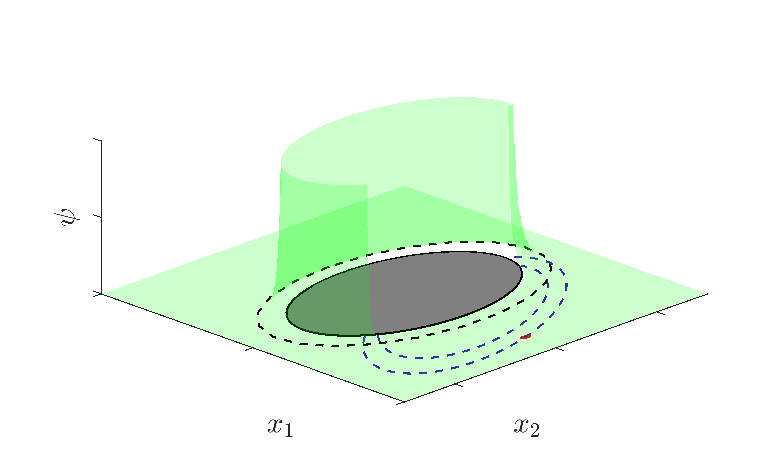}
        \caption{}
        \label{fig:3D psi vizualization}
    \end{subfigure}
 
    \caption{
(\subref{fig:2D psi vizualization}) Visualization of the penalty scaling function $\psi(d(x), s(x))$ with respect to $d(x)$ and $s(x)$. The left subfigure shows level sets of the distance function $d(x)$ (in blue) and the nominal control alignment $s(x)$ (in red), depicting their spatial variation around an obstacle. The right subfigure displays the corresponding values of $\psi(d(x), s(x))$, revealing its anisotropic behavior. As designed, $\psi(d(x), s(x)) = 0$ whenever $d(x) \ge \mu$ or $s(x) \ge \nu$, and grows unbounded as $d(x) \to 0$ with $s(x) \le 0$. In all other regions, $\psi$ takes smooth values in $(0, +\infty)$. This structure ensures that $\psi$ acts like a directional shield, activating only when the robot faces the obstacle-rather than surrounding it completely-effectively shaping the avoidance response as seen in subfigure (\subref{fig:3D psi vizualization}.)
}

    \label{fig:psi}
\end{figure}

Since $\psi(d, s)$ is a smooth function, the design parameters $\mu$ and $\nu$ define smooth transitions between the penalty-free and penalty-active regions, corresponding to distances in the range $(0, \mu)$ and directional projections in $(-\infty, \nu)$. To remain within the feasible set, the parameter $\mu$ must satisfy Assumptions~\ref{assumption:smoothBoundaries}--\ref{assumption:uniqeProjection}, which can be translated into the following condition:

\begin{equation}\label{condition:2}
   0<\mu<\min(h,\rho)-(R+\epsilon).
\end{equation}
Interestingly, the closed-form solution of the optimization problem defined in~\eqref{eq:optimization} is expressed as a smooth projection operation that acts linearly on the nominal control input.

\begin{proposition}\label{proposition:OP solution}
    Let $\mathcal{X}_\epsilon$ be the practical free space defined by~\eqref{eq:practicalfreespace}. The solution to the optimization problem~\eqref{eq:optimization} is given by:
    \begin{equation}\label{eq:smoothControl}
        u = \hat{\Pi}(x) \kappa_0(x),
    \end{equation}
    where the state-dependent projection matrix $\hat{\Pi}(x) \in \mathbb{R}^{n \times n}$ is defined as
    \begin{equation}\label{eq:projection matrix}
        \hat{\Pi}(x) = \mathbf{I}_n - \frac{\psi(d(x), s(x))}{1 + \psi(d(x), s(x))} \, \eta(x)\eta(x)^\top.
    \end{equation}
\end{proposition}
\begin{proof}
We consider the unconstrained optimization problem \eqref{eq:optimization}, which is strictly convex in $u$ due to the positive definiteness of the quadratic cost. Let \( J(u) \) denote the cost function:
\[
J(u) = \frac{1}{2}\|u - \kappa_0(x)\|^2 + \frac{1}{2}\psi(d(x), s(x))\left(\eta(x)^\top u\right)^2.
\]
This function is differentiable and strictly convex, hence it admits a unique minimizer. To determine the optimal control input, we compute the gradient of \( J(u) \) with respect to \( u \):
\begin{equation}
    \nabla J(u) = u - \kappa_0(x) + \psi(d(x), s(x)) \eta(x)\eta(x)^\top u.
\end{equation}
Setting \( \nabla J(u) = 0 \) yields the first-order optimality condition:
\begin{equation}\label{eq:gradientZero}
    u - \kappa_0(x) + \psi(d(x), s(x)) \eta(x)\eta(x)^\top u = 0.
\end{equation}

To solve for \( u \), we first project both sides of \eqref{eq:gradientZero} along \( \eta(x)^\top \). Using the normalization \( \|\eta(x)\| = 1 \), we obtain
\[
\eta(x)^\top u \left[ 1 + \psi(d(x), s(x)) \right] = \eta(x)^\top \kappa_0(x),
\]
which implies
\begin{equation}\label{eq:projU}
    \eta(x)^\top u = \frac{\eta(x)^\top \kappa_0(x)}{1 + \psi(d(x), s(x))}.
\end{equation}

Substituting \eqref{eq:projU} back into \eqref{eq:gradientZero}, we find
\[
u = \kappa_0(x) - \frac{\psi(d(x), s(x))}{1 + \psi(d(x), s(x))} \eta(x)\eta(x)^\top \kappa_0(x).
\]
The resulting expression corresponds to the closed-form solution stated in the proposition.
\end{proof}
According to Definition~\ref{def:Penalty function}, the penalty scaling function \(\psi(d(x), s(x))\) is nonnegative and increases as the robot approaches the obstacle or as the nominal controller points toward it. As a result, the term \(\frac{\psi(d(x), s(x))}{1 + \psi(d(x), s(x))}\) lies in the interval \([0, 1)\), and the projection matrix \(\hat{\Pi}(x)\) smoothly interpolates between the identity matrix and an orthogonal projection. In particular, when \(\psi(d(x), s(x)) \to 0\), which occurs when the robot is far from the obstacle and the nominal controller is directed away from it, we recover \(\hat{\Pi}(x) \to \mathbf{I}_n\), so the control input satisfies \(u = \kappa_0(x)\). In contrast, when \(\psi(d(x), s(x)) \to +\infty\), the matrix \(\hat{\Pi}(x)\) tends to \(\Pi(x) := \mathbf{I}_n - \eta(x)\eta(x)^\top\), corresponding to the orthogonal projection onto the hyperplane normal to \(\eta(x)\). In this case, the control input is forced to be orthogonal to the obstacle gradient direction, thereby ensuring collision avoidance. This construction enables a smooth and state-dependent transition between nominal tracking and obstacle avoidance, governed entirely by the penalty function \(\psi(d(x), s(x))\).

From a practical standpoint, the computation of potentially large values of $\psi$ can be avoided altogether by directly defining and evaluating the bounded ratio $\frac{\psi(d(x), s(x))}{1 + \psi(d(x), s(x))}$; see the following remark for further clarification.

\begin{rem}\label{remark:PsiChoice}
The choice of the function $\psi(d(x), s(x))$ can be arbitrary as long as it satisfies Definition~\ref{def:Penalty function}. In our work, we use the following penalty scaling function:
\begin{equation}\label{eq:psi}
    \psi(d(x), s(x)) = \frac{\phi_\mu(d(x)) \, \phi_\nu(s(x))}{1 - \phi_\mu(d(x)) \, \phi_\nu(s(x))},
\end{equation}
where
\begin{equation}
    \phi_\tau(z) := \begin{cases}
        0,  & z \in [\tau, +\infty),\\
        \gamma(z),  & z \in (0, \tau),\\
        1, & z \in (-\infty, 0],
    \end{cases}
\end{equation}
and $\phi_\tau(\cdot)$ is a smooth transition function parameterized by $\tau > 0$. The function $\gamma(\cdot) \in [0,1]$ can be any sufficiently differentiable function that ensures a smooth transition between $1$ and $0$ over the interval $(0, \tau)$. For instance, to obtain a continuously differentiable function $\phi_\tau(\cdot)$, one can choose $\gamma(\cdot)$ as a cubic polynomial satisfying the boundary conditions $\gamma(0) = 1$, $\gamma'(0) = 0$, $\gamma(\tau) = 0$, and $\gamma'(\tau) = 0$. This yields: $\gamma(z) = a + bz + cz^2 + dz^3$ where $a=1$, $b=0$, $c=-3/\tau^2$, $d=2/\tau^3$. The function $\phi_\mu(d(x))$ ensures smoothness with respect to distance (i.e., it regularizes the controller near $d(x) = 0$), while $\phi_\nu(s(x))$ ensures smoothness with respect to the directionality encoded in $s(x)$. 

\textit{Importantly, using this formulation, the bounded term in the projection matrix \eqref{eq:projection matrix}, becomes
\[
\frac{\psi(d(x), s(x))}{1 + \psi(d(x), s(x))}=\phi_\mu(d(x)) \, \phi_\nu(s(x)),
\]
thereby avoiding any numerical issues associated with the potentially large values of $\psi$.} Figure~\ref{fig:Psi d s} depicts an example of the functions $\psi(d(x), s(x))$, $\phi_\mu(d(x))$, and $\phi_\nu(s(x))$.
\end{rem}

\begin{figure}[h!]
    \centering

    \begin{subfigure}[b]{0.5\linewidth}
        \centering
        \includegraphics[width=\linewidth]{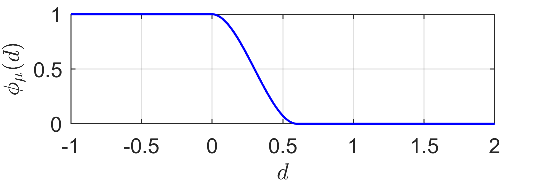}
        \caption{}
        \label{fig:d}
    \end{subfigure}
    \vfill
    \begin{subfigure}[b]{0.5\linewidth}
        \centering
        \includegraphics[width=\linewidth]{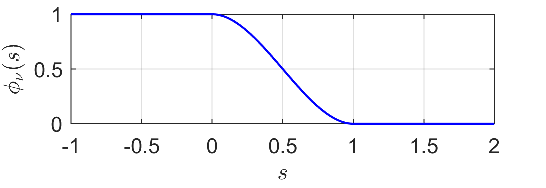}
        \caption{}
        \label{fig:s}
    \end{subfigure}

    \vfill
    \begin{subfigure}[b]{0.5\linewidth}
        \centering
        \includegraphics[width=\linewidth]{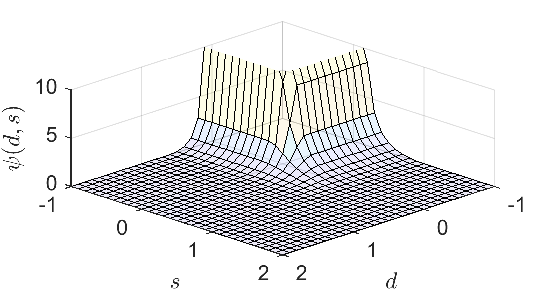}
        \caption{}
        \label{fig:psiofds}
    \end{subfigure}

    \caption{An example of the functions $\phi_\mu(d(x))$, $\phi_\nu(s(x))$ and $\psi(d(x),s(x))$. (\subref{fig:d}) represents $\phi_\mu(d(x))$, where $\mu=0.6$. (\subref{fig:s}) represents $\phi_\nu(s(x))$, where $\nu=1$. (\subref{fig:psiofds}) represents the $\psi(d(x),s(x))$, and shows that it blows up when $d(x)\to 0$ and $s(x)\to 0$ and vanishes when $d(x)\ge\mu$ or $s(x)\ge\nu$.}
    \label{fig:Psi d s}
\end{figure}

We recall our controller $u=\kappa(x)$ proposed in Proposition~\ref{proposition:OP solution} and defined by \eqref{eq:smoothControl}. From here, we can state our first result related to the safety of the robot under the dynamics \eqref{eq:dynamicalSystem}.

\begin{theorem}\label{theorem:invariance}
    Consider the set $\mathcal{X}\subset\mathbb R^n$ that describes the free space and satisfies Assumptions~\ref{assumption:smoothBoundaries}--\ref{assumption:uniqeProjection}. Consider the set $\mathcal{X}_\epsilon\in\mathbb R^n$ that describes the practical free space and is given by \eqref{eq:practicalfreespace}. Consider a nominal controller $\kappa_0(x)$ of class $\mathcal{C}^j$ and a penalty scaling function $\psi(\cdot)$ of class $\mathcal{C}^l$. Consider the dynamical system \eqref{eq:dynamicalSystem} under the control law \eqref{eq:smoothControl}. Then, the controller \eqref{eq:smoothControl} is of class $\mathcal{C}^{\min(j,k-1,l)}$ and its norm is always less than or equal to the norm of $\kappa_0(x)$. Moreover, if $\min(j,k-1,l)\ge 1$ then the closed-loop system admits a unique solution and the set $\mathcal{X}_\epsilon$ is positively invariant.
\end{theorem}
\begin{proof}
The smoothness of the controller \eqref{eq:smoothControl} depends directly on the classes of the functions $\psi$, $d$ and $\kappa_0$. According to Assumption~\ref{assumption:smoothBoundaries} the distance function $d(x)$ is of class $\mathcal{C}^k$ and as a result $\eta(x)\in\mathcal{C}^{k-1}$. Given the nominal controller $\kappa_0(x)$ is $\mathcal{C}^j$, then $s(x)=\kappa_0(x)^\top\eta(x)\in\mathcal{C}^{\min\{j,k-1\}}$. This implies that the composite function $\psi(d(x),s(x))$ is of class $\mathcal{C}^{\min\{j,k-1,l\}}$. Considering the function:
\begin{equation*}
    \alpha(x):=\frac{\psi(d(x),s(x))}{1+\psi(d(x),s(x))}.
\end{equation*}
Since the denominator $(1+\psi(d(x),s(x)))\neq 0$, then $\alpha(x)$ is the same class as $\psi(d(x),s(x))$. Thus $\alpha(x)\in\mathcal{C}^{\min\{j,k-1,l\}}$. Next, the outer production $\eta(x)\eta(x)^\top$ is $\mathcal{C}^{k-1}$, then $\alpha(x)\eta(x)\eta(x)^\top\in\mathcal{C}^{\min\{j,k-1,l\}}$. Therefore, $\kappa(x)$ is of class $\mathcal{C}^{\min\{j,k-1,l\}}$. 
To prove that the controller \eqref{eq:smoothControl} is bounded, we first expand its equation:
\begin{equation}
    \kappa(x) = \kappa_0(x)-\alpha(x)\eta(x)(\eta(x)^\top\kappa_0(x)).
\end{equation}
Then we can write $\|\kappa(x)\|^2$ as follows: 
\begin{align*}
\|\kappa(x)\|^2 &= \left\| \kappa_0(x) - \alpha(x) (\eta(x)^\top \kappa_0(x))\eta(x) \right\|^2 \\
&= \|\kappa_0(x)\|^2 - 2\alpha(x) (\eta(x)^\top \kappa_0(x))^2 \\ 
&+ \alpha(x)^2 (\eta(x)^\top \kappa_0(x))^2 \\
&= \|\kappa_0(x)\|^2 - \alpha(x)(2 - \alpha(x))(\eta(x)^\top \kappa_0(x))^2.
\end{align*}
Since \( \alpha(x) \in [0,1) \), we have \( \alpha(x)(2 - \alpha(x)) > 0 \). Thus,
\[
\|\kappa(x)\|^2 \le \|\kappa_0(x)\|^2 \quad \Rightarrow \quad \|\kappa(x)\| \le \|\kappa_0(x)\|.
\]
Hence, the norm $\|\kappa(x)\|$ of our controller is always less or equal than the norm of the nominal controller $\|\kappa_0(x)\|$.


Next, we prove that the closed-loop system admits a unique solution. Given that $\min\{j,k-1,l\}\ge 1$, then $\kappa(x)$ is at least continuously differentiable, hence locally Lipschitz on $\mathcal{X}_\epsilon$. Now, we show that any trajectory starting in $\mathcal{X}_\epsilon$ remains in it for all future time. It suffices to investigate the behavior at the boundary, {\it i.e.,} when $x\in\partial\mathcal{X}_\epsilon$, and when the $\eta(x)^\top\kappa_0(x)\le 0$. In this case, the nominal control is totally projected onto the hyperplane normal to $\eta(x)$:
\begin{align*}
        \dot{x}|_{d(x)=0,\eta(x)^\top\kappa_0\le 0} &= \kappa(x)|_{d(x)=0,\eta(x)^\top\kappa_0\le 0}\\
        &=(\mathbf{I}_n-\eta(x)\eta(x)^\top)\kappa_0(x).
\end{align*}
By multiplying both sides by $\eta(x)^\top$ and using the fact that $\eta(x)^\top\eta(x)=1$, we get
\begin{equation}
        \eta(x)^\top\dot{x}|_{d(x)=0,\eta(x)^\top\kappa_0\le 0} =0.
\end{equation}
This means that, in the worst case, {\it i.e.,} when $x$ is in the boundary $\partial\mathcal{X}_\epsilon$ and when the nominal controller $\kappa_0(x)$ is unsafe, the normal component of the robots velocity is null. Thus, the trajectories will stay inside or at the boundary of the practical free space $\mathcal{X}_{\epsilon}$. With this, we met the conditions of \cite[theorem 3.3]{khalil2002nonlinear} and as a result the closed-loop system admits a unique solution. Eventually, we showed that the set $\mathcal{X}_\epsilon$ is positively invariant under the proposed control law \eqref{eq:smoothControl} given the previous discussion.
\end{proof}

Theorem \ref{theorem:invariance} addresses the aspect of safety and suggests that safety is guaranteed for any nominal controller $\kappa_0(x)$ and for any shape of obstacles (convex or non-convex). Next, we study the motion-to-goal feature, {\it i.e.,} convergence of the robot's trajectories to the desired position $x_d$ under the proposed smooth controller $\kappa(x)$. This convergence is influenced by the choice of the nominal controller $\kappa_0(x)$. We specify a nominal controller based on a gradient-decent strategy, which we define as follows:
\begin{equation}\label{eq:nominalControl}
    \kappa_0(x) = -\nabla V(x),
\end{equation}
where $V:\mathbb{R}^n\to\mathbb{R}$ is a scalar-valued potential function and $\nabla V(x)$ is its gradient with respect to $x$. We impose the following assumption on $V(x)$:
\begin{assumption}\label{assumption:potential}
    Let $V(x)$ be a $\mathcal{C}^{j+1}$-class scalar function satisfying the following properties:
    \begin{enumerate}
        \item $V(x)$ is positive definite, i.e., $V(x_d)=0$ and $V(x)>0$ for all $x\neq x_d$.
        \item $V(x)$ is radially unbounded, i.e., $V(x)\to\infty$ as $\|x\|\to\infty$.
    \end{enumerate}
\end{assumption}
Under Assumption \ref{assumption:potential}, the nominal controller $\kappa_0(x)$ drives the robot to $x_d$ in the absence of the obstacles. In what follows, we establish our results on the stability properties of the closed-loop system \eqref{eq:closedLoopsystem} under the controller \eqref{eq:smoothControl}. 
\begin{figure}
    \centering
    \includegraphics[width=0.6\linewidth]{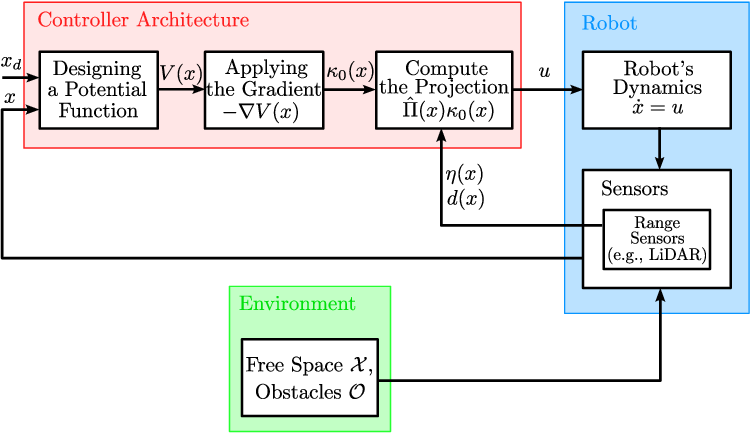}
    \caption{Block diagram illustrating the feedback avoidance controller proposed Theorem \ref{theorem:convergence}. The proposed simple control strategy in \eqref{eq:smoothControl} smoothly projects the nominal gradient descent controller \eqref{eq:nominalControl} onto the tangent space to the obstacle's boundary as the robot moves closer to the obstacle.  
    }
    \label{fig:diagram}
\end{figure}

\begin{theorem}\label{theorem:convergence}
    Consider the set $\mathcal{X}\subset\mathbb R^n$ that describes the free space and satisfies Assumptions~\ref{assumption:smoothBoundaries}--\ref{assumption:uniqeProjection}. Consider the set $\mathcal{X}_\epsilon\in\mathbb R^n$ that describes the practical free space and is given by \eqref{eq:practicalfreespace}. Consider the dynamical system \eqref{eq:dynamicalSystem} under the $\mathcal{C}^{\min(j,k-1,l)}$-class control law \eqref{eq:smoothControl}, with $\kappa_0(x)$ as in \eqref{eq:nominalControl}. Then, the potential $V(x)$ is non-increasing and trajectories converge to the set $\mathcal{E}\cup\{x_d\}$, where
    \begin{equation}\label{eq:E}
        \mathcal{E}:=\{x\in\partial\mathcal{X}_\epsilon:\nabla V(x)=\lambda\eta(x), \lambda\in\mathbb{R}_{>0}\}
        \end{equation}
        is a set of measure zero.
\end{theorem}
\begin{proof}
We start by letting $V(x)$ be a positive definite function that satisfies Assumption \ref{assumption:potential}. Its time derivative in view of the dynamical system \eqref{eq:dynamicalSystem} and under the controller \eqref{eq:smoothControl} is given by
\begin{align*}
    \dot{V}(x)&=\nabla V(x)^\top\dot{x}\nonumber\\
    &=-\nabla V(x)^\top \hat{\Pi}(x)\nabla V(x).
\end{align*}
The term  $\nabla V(x)^\top \hat{\Pi}(x)\nabla V(x)$ is positive semi-definite as shown in the following:
\begin{align*}
    &\nabla V(x)^\top \hat{\Pi}(x)\nabla V(x)\\
    &=\nabla V(x)^\top\Big[\mathbf{I}_n-\frac{\psi(d(x),s(x))}{1+\psi(d(x),s(x))}\eta(x)\eta(x)^\top\Big]\nabla V(x)\\
    &=\|\nabla V(x)\|^2-\frac{\psi(d(x),s(x))}{1+\psi(d(x),s(x))}\|\nabla V(x)\|^2\cos^2{\theta}(x)\\
    &=\|\nabla V(x)\|^2\Big[1-\frac{\psi(d(x),s(x))}{1+\psi(d(x),s(x))}\cos^2{\theta}(x)\Big]\ge 0,
\end{align*}
since $\frac{\psi(d(x),s(x))}{1+\psi(d(x),s(x))}\in[0,1]$, where $\theta(x)$ is the angle  between the two vectors $\nabla V(x)$ and $\eta(x)$. Therefore, we have
\begin{equation}
    \dot{V}(x)\le 0,\quad\forall x\in\mathcal{X}_{\epsilon}.
\end{equation}
This implies two things. First, $\frac{\psi(d(x),s(x))}{1+\psi(d(x),s(x))}=1$ which can be attained only when $\psi(d(x),s(x))\to+\infty$. In other terms, according to the Definition~\ref{def:Penalty function} of the Penalty Scaling Function, this condition is satisfied for any point $x$ that lies on the boundary $\partial\mathcal{X}_\epsilon$ and satisfies $\kappa_0(x)^\top\eta(x)\le 0$. Second, $\cos^2{\theta}(x)=1$ which is met for all points $x$ such that $\nabla V(x)=\lambda\eta(x)$, where $\lambda>0$. The sign of the parameter $\lambda$ is deduced from the fact that $\kappa_0(x)^\top\eta(x)\le 0$. Therefore, the set $\{x_d\}\cup\mathcal{E}$, where $\mathcal{E}$ is defined in \eqref{eq:E}, represent all the points for which $\dot{V}(x)=0$. Moreover, this set is the largest invariant set contained in $\{x\in\mathcal{X}_\epsilon: \dot{V}(x)=0\}$, since any point $x\in\{x_d\}\cup\mathcal{E}$, when substituted into the closed-loop system dynamics \eqref{eq:closedLoopsystem} with the control law \eqref{eq:smoothControl}, satisfies $\dot{x}=0$. Finally, according to LaSalle's Theorem, all the solutions starting in $\mathcal{X}_\epsilon$ converge asymptotically to the set of points $\{x_d\}\cup\mathcal{E}$.
Since $\partial\mathcal{X}_\epsilon$ is a set of measure zero, and $\mathcal{E}$ is a subset of $\partial\mathcal{X}_\epsilon$, it follows that $\mathcal{E}$ is also of measure zero.
\end{proof}

\begin{figure}[h!]
    \centering
    \includegraphics[width=0.5\linewidth]{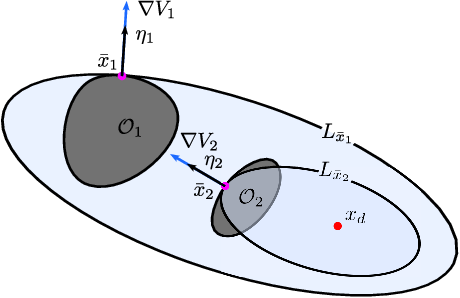}
    \caption{An illustration of the undesired equilibria in case of 2D obstacles. The (gray) regions $\mathcal{O}_1$ and $\mathcal{O}_2$ represent distinct obstacles, each associated with an undesired equilibrium point $\bar{x}_1$ and $\bar{x}_2$, respectively. The sets $L_{\bar{x}_1}$ and $L_{\bar{x}_2}$ are the level sets passing through these equilibria. At each $\bar{x}_i$, the gradient $\nabla V_1$ normal to the level set, collinear and point in the same direction with $\eta_1$, the normal to the boundary of the obstacle.
   }
    \label{fig:undesiredPoints}
\end{figure}

Theorem \ref{theorem:convergence} address the problem of motion-to-goal of the robot under the proposed controller \eqref{eq:smoothControl}. The set $\mathcal{E}$ represents the set of undesired equilibria, which are points on the boundary of the practical free space $\partial\mathcal{X}_\epsilon$ where the gradient $\nabla V(x)$ becomes collinear with the normal $\eta(x)$ and points in the same direction. To illustrate this, we start by defining the level-set 
\begin{equation}\label{eq:levelset}
    L_{\bar{x}}:=\{x\in\mathbb{R}^n:V(x)=V(\bar{x})\},
\end{equation}
where $\bar{x}\in\mathcal{E}$. The outward normal to $L_{\bar{x}}$ at $x=\bar{x}$ is $\nabla V(\bar{x})/\|\nabla V(\bar{x})\|$. This implies that the boundary of the dilated obstacle set by the parameter $R+\epsilon$ share the same point $\bar{x}$ with the set $L_{\bar{x}}$ as well as the same unit normal vector (See Fig. \ref{fig:undesiredPoints}). In other words, if $\mathcal{O}_\epsilon$ is the dilated obstacle, then the boundary $\partial\mathcal{O}_\epsilon$ have the same tangent space with $L_{\bar{x}}$ at $x=\bar{x}$, i.e., $T_{\bar{x}}\partial\mathcal{O}_\epsilon=T_{\bar{x}}L_{\bar{x}}$, and any tangent vector $v(\bar{x})$ at $\bar{x}$ lies in both tangent spaces. This allows us to study the local relative curvature between the two hypersurfaces $\partial\mathcal{O}_\epsilon$ and $L_{\bar{x}}$ near $\bar{x}$. For this purpose, we define the normal curvature as follows:
\begin{definition}[Normal Curvature ]\cite[Definitions 2.2-2.3, Chap 5]{ShapeOperator}\label{def:normal curvature}
    Consider the $(n-1)$-dimensional hypersurface $M\subset\mathbb R^n$ and its tangent space $T_pM$ at $p\in M$. Let $v_p\in T_pM$ be a unit vector. The normal curvature\footnote{The sign of the normal curvature depends on the choice of the hypersurface normal
    . In this work, we adopt the convention where the normal to the hypersurface defining the boundary of a set is oriented to the interior of the said set, \textit{i.e.,} an inward normal. } $C:M\to\mathbb R$ in the direction $v_p$ is defined as:
    \begin{equation}
        C(v_p):=-v_p^\top \nabla N_p v_p,
    \end{equation}
    where $N_p$ is the inward unit normal to $M$ at $p$.
\end{definition}
The notion of normal curvature, as defined above, can be directly related to the second-order properties of the scalar functions that generate various hypersurfaces, such as the dilated obstacle boundaries $\partial\mathcal{O}_\epsilon$ and the level set $L_{\bar{x}}$ introduced in \eqref{eq:levelset}. These geometric objects will play a central role in the subsequent analysis. To be consistent with the adopted normal orientation convention, we take the inward normal to $\partial\mathcal{O}_\epsilon$ as $-\eta(x)$ and the inward normal to $L_{\bar{x}}$ as $-\nabla V(x)$. This yields the following expressions for the normal curvature.
The normal curvature of the boundary $\partial\mathcal{O}_\epsilon$ at $\bar{x}$ in the direction $v(\bar{x}) \in T_{\bar{x}}\partial\mathcal{O}_\epsilon$ is given by
\begin{equation}\label{eq:curvature-obstacle}
    C_{\partial\mathcal{O}_\epsilon}(v(\bar{x})) = v(\bar{x})^\top\mathbf{H}_d(\bar{x})v(\bar{x}),
\end{equation}
where $\mathbf{H}_d(\bar{x})$ denotes the Hessian of the distance function $d(x)$ evaluated at $\bar{x}$.  
Similarly, the normal curvature of the level set $L_{\bar{x}}$ at the same point $\bar{x}$ in the direction $v(\bar{x}) \in T_{\bar{x}}L_{\bar{x}}$ is given by
\begin{equation}\label{eq:curvature-levelset}
    C_{L_{\bar{x}}}(v(\bar{x})) = \frac{v(\bar{x})^\top\mathbf{H}_V(\bar{x})v(\bar{x})}{\|\nabla V(\bar{x})\|},
\end{equation}
where $\mathbf{H}_V(\bar{x})$ is the Hessian of the potential function $V(x)$ evaluated at $\bar{x}$.

In the remainder of this section, we will analyze the nature of undesired equilibria by comparing these two curvature values. Notably, the normal curvature also plays a key role in determining whether undesired equilibria are isolated. The next proposition provides a condition that ensures isolated equilibria.
.
\begin{proposition}\label{proposition:isolatedEq}
    Consider the set $\mathcal{X}\subset\mathbb R^n$ that describes the free space and satisfies Assumptions~\ref{assumption:smoothBoundaries}--\ref{assumption:uniqeProjection}. Consider the set $\mathcal{X}_\epsilon\in\mathbb R^n$ that describes the practical free space and is given by \eqref{eq:practicalfreespace}. Consider the dynamical system \eqref{eq:dynamicalSystem} under the $\mathcal{C}^{\min(j,k-1,l)}$-control law \eqref{eq:smoothControl}, with $\kappa_0(x)$ as in \eqref{eq:nominalControl}. Then, the undesired equilibria  in $\mathcal{E}$ are isolated if, for any $\bar{x}\in\mathcal{E}$, all directions $v(\bar{x})\in T_{\bar{x}}\partial\mathcal{O}$ satisfies the condition
    \begin{equation}\label{eq:CurvatureConditionIsolated}
             C_{\partial\mathcal{O}_\epsilon}(v(\bar{x}))\neq C_{L_{\bar{x}}}(v(\bar{x})),
    \end{equation}
    where $C_{\partial\mathcal{O}_\epsilon}(v(\bar{x}))$ and $C_{L_{\bar{x}}}(v(\bar{x}))$ are the normal curvature of $\partial\mathcal{O}_\epsilon$ and $L_{\bar{x}}$, respectively, evaluated at $x=\bar{x}$ in the direction $v(\bar{x})$.
\end{proposition}
\begin{proof}
We begin by considering a neighborhood of the undesired equilibrium point $\bar{x} \in \mathcal{E}$, specifically the ball $\mathcal{B}(\bar{x}, r)$. Let $\mathcal{P} = \{x : \nabla V(x)^\top \eta(x) \ge 0\}$ denote the region where the avoidance controller is active. Within the intersection $\mathcal{B}(\bar{x}, r) \cap \mathcal{P} \cap \partial \mathcal{X}_\epsilon$, the system evolves according to the projected dynamics
\begin{equation}\label{eq:localDynamics}
    \dot{x} = -(\mathbf{I}_n - \eta(x)\eta(x)^\top)\nabla V(x),
\end{equation}
which ensures that the motion remains tangent to the boundary, as $\eta(x)^\top \dot{x} = 0$. To analyze the equilibria, we restrict the dynamics to the tangent space at $\bar{x}$,
$
    T_{\bar{x}} := \{v \in \mathbb{R}^n : \eta(\bar{x})^\top v = 0\},
$
and consider perturbations of the form
$
    x = \bar{x} + \sigma \delta x,
$
where $\delta x \in T_{\bar{x}}$ is a unit vector and $\sigma > 0$ is arbitrarily small. Linearizing \eqref{eq:localDynamics} at $\bar{x}$ gives
\begin{equation}\label{eq:LinearizedReducedLocalDynamics}
    \dot{\delta x} = \sigma J(\bar{x}) \delta x,
\end{equation}
where $J(\bar{x})$ denotes the Jacobian of the vector field in \eqref{eq:localDynamics}. Since the motion is restricted to $T_{\bar{x}}$, the effective dynamics evolve in an $(n - 1)$-dimensional subspace. To ensure that $\bar{x}$ is an isolated equilibrium of \eqref{eq:localDynamics}, it suffices to show that the Jacobian restricted to $T_{\bar{x}}$ admits no purely imaginary eigenvalues \footnote{Note that $J(\bar{x})$ admits a zero eigenvalue along $\eta(\bar{x})$ but this direction is orthogonal to the constraint manifold and does not influence the reduced dynamics.}. By the Hartman-Grobman Theorem \cite{Perko2001}, this guarantees that $\bar{x}$ remains the unique equilibrium in a sufficiently small neighborhood.

The Jacobian $J(x)$ is computed as
\begin{multline}\label{eq:jacobian}
    J(x) = -\mathbf{H}_V(x) + (\eta(x)^\top \nabla V(x)) \mathbf{H}_d(x)\\ 
    + \eta(x)\left[\nabla V(x)^\top \mathbf{H}_d(x) + \eta(x)^\top \mathbf{H}_V(x)\right],
\end{multline}
where $\mathbf{H}_V(x)$ and $\mathbf{H}_d(x)$ are the Hessians of $V(x)$ and $d(x)$, respectively.
At an undesired equilibrium $\bar{x} \in \mathcal{E}$, the potential gradient satisfies $\nabla V(\bar{x}) = \lambda \eta(\bar{x})$ for some scalar $\lambda > 0$. Substituting into \eqref{eq:jacobian}, we obtain
\begin{multline}\label{eq:jacobianUndesiredEq}
    J(\bar{x}) = \lambda \mathbf{H}_d(\bar{x}) - \mathbf{H}_V(\bar{x}) + \\\eta(\bar{x})\eta(\bar{x})^\top \left(\lambda \mathbf{H}_d(\bar{x}) + \mathbf{H}_V(\bar{x})\right).
\end{multline}

Projecting the linearized dynamics \eqref{eq:LinearizedReducedLocalDynamics} onto the tangent direction $\delta x \in T_{\bar{x}}$, we obtain in view of \eqref{eq:curvature-obstacle}-\eqref{eq:curvature-levelset}
\begin{align}
    \delta x^\top \dot{\delta x} &= \delta x^\top J(\bar{x}) \delta x \\
    &= \delta x^\top (\lambda \mathbf{H}_d(\bar{x}) - \mathbf{H}_V(\bar{x})) \delta x \\
    &= \lambda C_{\partial \mathcal{O}_\epsilon}(\delta x) - \lambda C_{L_{\bar{x}}}(\delta x).
\end{align}
Since the matrix $\lambda \mathbf{H}_d(\bar{x}) - \mathbf{H}_V(\bar{x})$ is symmetric, this expression defines a real-valued quadratic form on $T_{\bar{x}}$. If 
\[
C_{\partial \mathcal{O}_\epsilon}(\delta x) \neq C_{L_{\bar{x}}}(\delta x), \quad \forall\, \delta x \in T_{\bar{x}},\; \|\delta x\| = 1,
\]
then the quadratic form is nonzero in all directions of the tangent space, implying that $J(\bar{x})$ admits no purely imaginary eigenvalues on $T_{\bar{x}}$. Therefore, the equilibrium $\bar{x}$ is isolated.
\end{proof}

The previous result established that undesired equilibria are isolated whenever the curvatures of the obstacle boundary and the potential level set differ in all tangent directions. We next examine the stability properties of both the desired and undesired equilibria. In contrast to the previous result on isolation, the condition used here involves a curvature inequality that need not hold in all tangent directions. This distinction becomes particularly relevant in higher-dimensional settings and will be discussed in detail after the theorem.

\begin{theorem}\label{theorem:stability}
    Consider the set $\mathcal{X}\subset\mathbb R^n$ that describes the free space and satisfies Assumptions~\ref{assumption:smoothBoundaries}--\ref{assumption:uniqeProjection}. Consider the set $\mathcal{X}_\epsilon\in\mathbb R^n$ that describes the practical free space and is given by \eqref{eq:practicalfreespace}. Consider the dynamical system \eqref{eq:dynamicalSystem} under the $\mathcal{C}^{\min(j,k-1,l)}$-control law \eqref{eq:smoothControl}, with $\kappa_0(x)$ as in \eqref{eq:nominalControl}. If there exist a direction $v(\bar{x})\in T_{\bar{x}}\partial\mathcal{O}$ such that, for all $\bar{x}\in\mathcal{E}$, we have
    \begin{equation}\label{eq:CurvatureCondition}
        C_{\partial\mathcal{O}_\epsilon}(v(\bar{x})) > C_{L_{\bar{x}}}(v(\bar{x})),
    \end{equation}
    then,
    \begin{enumerate}
    \item all the undesired equilibria $\bar x\in\mathcal{E}$ are unstable, and
    \item the desired equilibrium $x_d$ almost globally asymptotically stable.  
    \end{enumerate}
\end{theorem}
\begin{proof}
To prove {\it Item 1)}, we recall the local dynamics \eqref{eq:localDynamics} restricted to the set $\mathcal{P}=\{x:\nabla V(x)^\top\eta(x)\ge 0\}$ . We consider the following positive definite function: 
\begin{equation}
    W=\frac{1}{2}\|x-\bar{x}\|^2,
\end{equation}
where its time derivative, in view of the local dynamics \eqref{eq:localDynamics}, is given by:
\begin{equation}
     \dot{W}=-(x-\bar x)^\top(\mathbf{I}_n-\eta(x)\eta(x)^\top)\nabla V(x).\label{eq:wDot}
\end{equation}

We consider an arbitrary close point to $\bar{x}$ given by
\begin{equation}
    x=\bar{x}+\sigma\delta x,
\end{equation}
where $\delta x$ is a unite vector and $\sigma$ is an arbitrary small positive real. Using \eqref{eq:E}, we can write the equation of $\nabla V(x)$ as follows
\begin{align}\label{eq:perturbation}
    \nabla V(x)&=\nabla V(\bar{x})+\sigma\mathbf{H}_V(\bar{x})\delta x,\nonumber\\
    \nabla V(x)&= \lambda\eta(\bar{x})+\sigma\mathbf{H}_V(\bar{x})\delta x,
\end{align}

The gradient $\eta(x)$ can be approximated as follows
\begin{equation}\label{eq:gradientApprox}
    \eta(x)=\eta(\bar x)+\sigma\mathbf{H}_d(\bar x)\delta x+O(\sigma^2).
\end{equation}

We can rewrite (\ref{eq:wDot}) by replacing with (\ref{eq:perturbation}) and (\ref{eq:gradientApprox}), which results in the following:
\begin{multline}\label{eq:wDotperturbed}
    \dot{W}=-\sigma\delta x^\top[ 
    \Pi(\bar x)-\\ \sigma\eta(\bar x)\delta x^\top\mathbf{H}_d(\bar x)-\sigma\mathbf{H}_d(\bar x)\delta x\eta(\bar x)^\top-\\ \sigma^2\mathbf{H}_d(\bar x)\delta x(\mathbf{H}_d(\bar x)\delta x)^\top
    ](\lambda\eta(\bar x)+\sigma\mathbf{H}_V(\bar{x})\delta x).
\end{multline}

We ignore third order terms or higher. We cancel the term $\Pi(\bar x)\eta(\bar x)$ since $\eta(\bar x)$ is orthogonal to the tangent hyperplane defined at $\bar x$. We can approximate the distance function as follows:
\begin{equation}
    \mathbf{d}(x)=\mathbf{d}(\bar x)+\sigma\eta(\bar x)^\top\delta x,
\end{equation}
Since $x\in\mathcal{B}\cap\mathcal{P}\cap\partial\mathcal{X}_{\epsilon}$, then $\mathbf{d}(x)=\mathbf{d}(\bar x)=0$. Therefore, $\eta(\bar x)^\top\delta x=0$. After simplification, the time derivative of the Lyapunov-like function takes the form
\begin{equation}\label{eq:wDotFinal}
    \dot{W} = \sigma^2 \delta x^\top (\lambda \mathbf{H}_d(\bar{x}) - \mathbf{H}_V(\bar{x})) \delta x.
\end{equation}
Using the expressions of the normal curvatures in \eqref{eq:curvature-obstacle}-\eqref{eq:curvature-levelset}, this can be rewritten as
\begin{equation}
    \dot{W} = \sigma^2 \lambda \left( C_{\partial \mathcal{O}_\epsilon}(\delta x) - C_{L_{\bar{x}}}(\delta x) \right).
\end{equation}
Therefore, if there exists a direction \( v(\bar{x}) \in T_{\bar{x}} \partial \mathcal{O}_\epsilon \) such that
\begin{equation}
    C_{\partial \mathcal{O}_\epsilon}(v(\bar{x})) > C_{L_{\bar{x}}}(v(\bar{x})),
\end{equation}
then \( \dot{W} > 0 \) along that perturbation direction, and the equilibrium point \( \bar{x} \) is unstable.

For {\it Item 2)}, we prove that the basin of attraction of the undesired equilibria is a set of measure zero. We start by recalling the equation of \eqref{eq:jacobianUndesiredEq} evaluated at $\bar{x}$ of $\kappa(x)$ restricted to the set $\mathcal{B}\cap\mathcal{P}\cap\partial\mathcal{X}_\epsilon$. As shown previously, if the curvature condition is satisfied, then there exist a direction $\delta x$ such that $\eta(\bar{x})^\top\delta x=0$ and $\delta x^\top (\lambda\mathbf{H}_d(\bar{x})-\mathbf{H}_V(\bar{x})) \delta x>0$. Therefore, we have that
\begin{equation}
    \delta x^\top J(x) \delta x>0.
\end{equation}
Since the Jacobian $J(x)$ as computed in \eqref{eq:jacobian} is a symmetric matrix, then it admits at least one positive eigenvalue. We denote by $\phi_t$ the flow of the closed-loop dynamical system \eqref{eq:closedLoopsystem}, and the stable manifold $\mathcal{S}$ for each undesired equilibrium point which satisfies
    \begin{equation}
        \lim_{t\to\infty}\phi_t(c)=\bar x,\hspace{4mm} \forall c\in\mathcal{S},\hspace{4mm}\text{where }\bar x\in\mathcal{E}. 
    \end{equation}
If the curvature condition is satisfied, then the Jacobian of $\kappa(x)$ evaluated locally at a point $\bar x\in\mathcal{E}$ admits at least one positive eigenvalue. Hence, the stable manifold $\mathcal{S}$ is at most $(n-1)$-dimensional manifold \cite[The Stable Manifold Theorem, Pg 107]{Perko2001}, and as a result, it is measure zero in the $n$-dimensional space. Since the closed-loop system \eqref{eq:closedLoopsystem} admits unique solutions, the global stable manifold at $\bar x\in\mathcal{E}$, defined as \cite[Definition 3, Pg 113]{Perko2001}
\begin{equation}
        W^s(\bar x)=\bigcup_{t\le0}\phi_t(\mathcal{S}),
\end{equation}
is also a measure zero set.
\end{proof}

\begin{figure}[h!]
    \centering
    \includegraphics[width=0.5\linewidth]{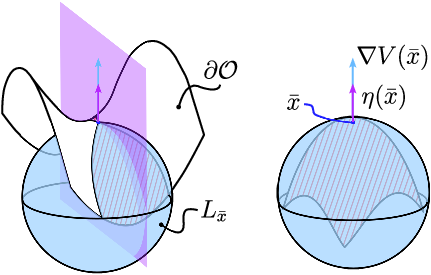}
    \caption{An illustration showcasing the curvature condition \eqref{eq:CurvatureCondition} for 3D obstacles. The curvature condition can be satisfied for convex obstacles (right) or non-convex obstacles such as saddle-shaped obstacles (left) provided that there exist a direction for which the obstacle's boundary sharply curved than the level set $L_{\bar{x}}$.}
    \label{fig:curvature3D}
\end{figure}

Theorem~\ref{theorem:stability} tackles the stability of the desired goal under the dynamical system~\eqref{eq:dynamicalSystem} with the controller~\eqref{eq:smoothControl}. Similarly to Proposition~\ref{proposition:isolatedEq}, the theorem employs a geometric criterion, stated in terms of the normal curvature, to characterize the instability of the undesired equilibria and the AGAS of the desired goal. 

In a 2D environment, an obstacle boundary at a given point has only one tangential direction. Consequently, if the instability condition~\eqref{eq:CurvatureCondition} holds at that point, then the isolation condition~\eqref{eq:CurvatureConditionIsolated} is also automatically satisfied, and all the undesired equilibria are isolated and unstable. In 3D, on the other hand, the tangent space is two-dimensional, and the curvature conditions must be verified in all directions lying in that plane. Therefore, it is possible for the instability condition~\eqref{eq:CurvatureCondition} to hold in a given direction, ensuring local instability, while the isolation condition~\eqref{eq:CurvatureConditionIsolated} fails in another, potentially admitting a continuum of undesired equilibria. In such cases, the undesired equilibria may be unstable yet not isolated. This scenario does not contradict the AGAS result, since the set of initial conditions converging to such equilibria remains of measure zero. 

Overall, Theorem~\ref{theorem:stability} establishes that if the boundary of the obstacle is locally more curved than the level set of the potential function, in at least one tangential direction at each undesired equilibrium, then the instability of all such equilibria is guaranteed. As a result, complex geometric configurations in 3D can be employed while preserving stability guarantees, as illustrated in Figure \ref{fig:curvature3D}.

\begin{rem}\label{remark:multObs}
    A natural extension of our framework consists in incorporating multiple obstacles into the unconstrained optimization problem~\eqref{eq:optimization} by summing several penalty terms. While the original formulation considers only the {\it closest obstacle}, we can generalize the approach by assigning to each sufficiently close obstacle~$\mathcal{O}_i$ a distance function~$d_i(x)$ and an associated unit normal~$\eta_i(x)$. The resulting optimization problem becomes:
    \begin{equation}\label{eq:optimizationMultipleObstacles}
        \min_u \, \|u - \kappa_0(x)\|^2 + \sum_i \psi(d_i(x), s_i(x))(u^\top \eta_i(x))^2,
    \end{equation}
    where $s_i(x) := \kappa_0(x)^\top \eta_i(x)$. However, the theoretical guarantees of almost global asymptotic stability (AGAS) established in Theorem~\ref{theorem:stability} no longer hold under this formulation, as the inclusion of multiple penalty terms may introduce undesired local minima. 

    One possible remedy is to adopt an adaptive strategy that selectively activates a subset of the penalty terms by considering only the closest obstacles, thus progressively reducing~\eqref{eq:optimizationMultipleObstacles} to the single-obstacle case, where AGAS is recovered. Nevertheless, such a strategy requires a detection mechanism capable of identifying the presence of local minima and triggering appropriate adjustments to the active penalty set. Developing such mechanisms is an important direction for future work.
\end{rem}

\begin{rem}[Instability of undesired equilibria]\label{remark:curvature}
The instability condition in Theorem~\ref{theorem:stability} relies on a comparison between the normal curvature of the obstacle boundary and that of the level sets of the potential function \(V(x)\). While our framework accommodates general smooth potentials, this condition may fail to hold for a given nominal choice, leading to undesired stable equilibria. In such cases, one can intentionally adjust the potential to reduce the curvature of its level sets and recover instability.

For instance, consider the case where the nominal controller is derived from a quadratic potential of the form
\[
V(x) = \frac{1}{2}(x - x_d)^\top P (x - x_d),
\]
with \(P \succ 0\). The curvature of the level sets is directly determined by the eigenstructure of \(P\). By selecting \(P\) with smaller eigenvalues along directions tangent to the obstacle boundary at an undesired equilibrium, one can effectively flatten the level sets locally and ensure that the curvature of the obstacle dominates, thereby satisfying the instability condition~\eqref{eq:CurvatureCondition}.

However, this strategy presumes the ability to identify situations in which the curvature condition is not satisfied. In practice, this can be addressed through an \emph{online detection mechanism} that monitors local geometric cues, such as the alignment between the gradient \(\nabla V(x)\) and the obstacle normal \(\eta(x)\), or the absence of divergence in the closed-loop flow near the boundary. Upon detection of a potential tangency or convergence to an undesired equilibrium, the metric \(P\) (in the case of quadratic potentials) can be locally adjusted to reduce the level set curvature and restore the instability condition. The development of such adaptive shaping mechanisms remains an interesting direction for future research.
\end{rem}

\begin{figure}[t]
\centering
    \includegraphics[trim={1.8cm 0.4cm 1.2cm 0.95cm},clip,width=\textwidth]{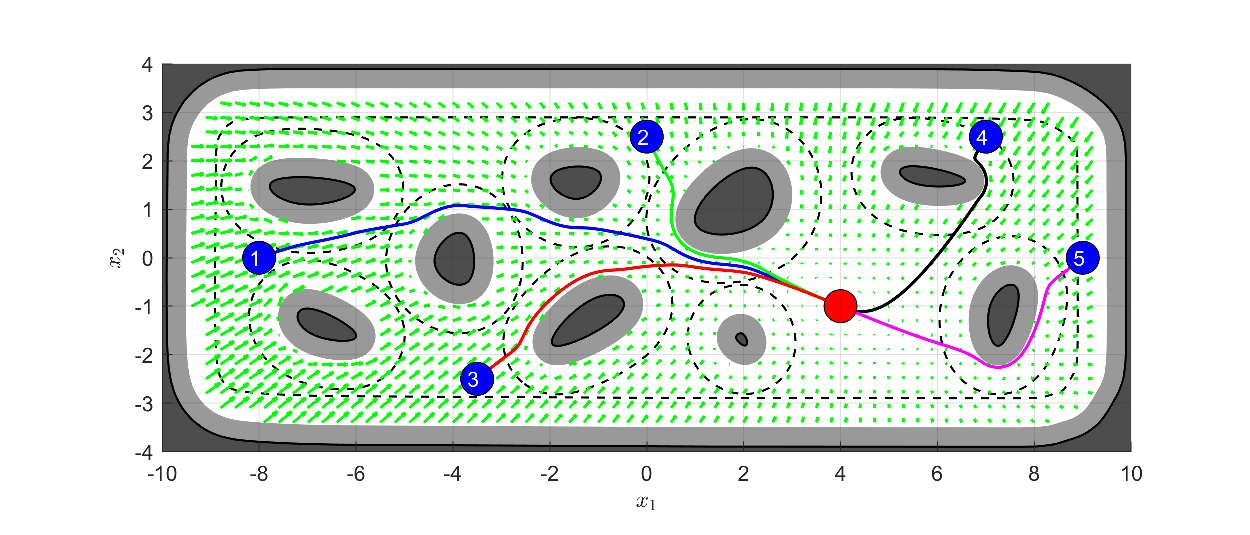}
    \caption{The trajectories of the robot in a 2D environment starting from a set of initial positions (blue) away from the goal (red) while avoiding the obstacles (dark gray). The (light gray) regions are the dilation of the obstacles (dark gray) by the distance $R+\epsilon$. The (black) dashed lines are a dilation by the geometric parameters $\mu$. The (colored) paths represent the resulting trajectories under our approach and the (green) arrows represent the corresponding vector field.
    See video: \href{https://youtu.be/3WpzqI5blOg}{https://youtu.be/3WpzqI5blOg} }
    \label{fig:2Dsimulation}
\end{figure}

\section{Simulation  results}\label{section:Simulation results}

\begin{figure}[h!]
    \centering

    \begin{subfigure}[b]{0.5\linewidth}
        \centering
        \includegraphics[width=\linewidth]{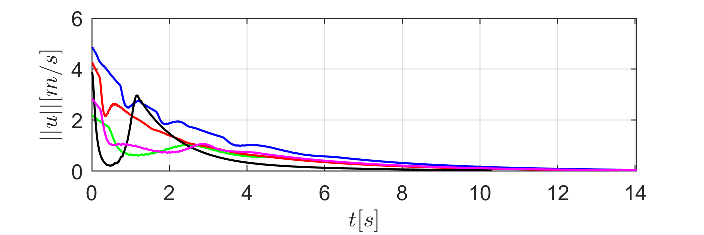}
        \caption{}
        \label{fig:u}
    \end{subfigure}
    \vfill
    \begin{subfigure}[b]{0.5\linewidth}
        \centering
        \includegraphics[width=\linewidth]{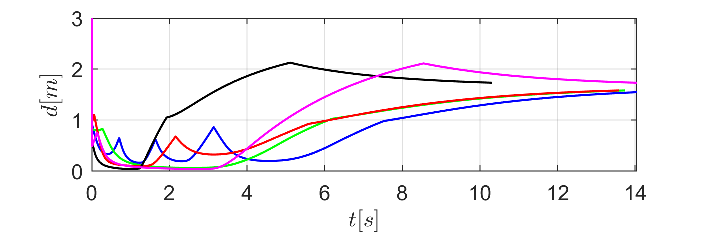}
        \caption{}
        \label{fig:d1}
    \end{subfigure}

    \vfill
    \begin{subfigure}[b]{0.5\linewidth}
        \centering
        \includegraphics[width=\linewidth]{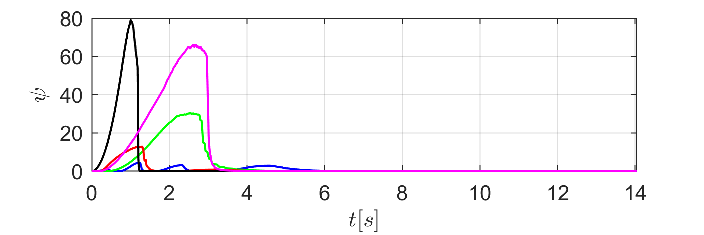}
        \caption{}
        \label{fig:psi1}
    \end{subfigure}

    \caption{Time evolution of (\subref{fig:u}) the control input, (\subref{fig:d1}) the distance to the obstacles and (\subref{fig:psi1}) the penalty scaling function for the 2D case and the initial conditions: 1 in (blue), 2 in (green), 3 in (red), 4 in (black) and 5 in (magenta).}
    \label{fig:u d psi}
\end{figure}

In this section, we demonstrate the performance of the proposed smooth controller in guaranteeing safe navigation in unknown environments by running both 2D and 3D simulations. The objective is to validate the theoretical results derived in previous sections, particularly the safety and convergence guarantees provided by the control law.

For this purpose, we consider the following quadratic positive definite function:
\begin{equation}
V(x) = \frac{1}{2}(x - x_d)^\top \mathbf{P}(x - x_d),
\end{equation}
where $\mathbf{P} \in \mathbb{R}^{n \times n}$ is a positive definite matrix. This leads to the following expression for the nominal controller introduced in \eqref{eq:nominalControl}:
\begin{equation}
\kappa_0(x) = -\mathbf{P}(x - x_d).
\end{equation}
As previously introduced in Remark~\ref{remark:PsiChoice}, the penalty scaling function defined by \eqref{eq:psi} is employed in this numerical application. To model limited onboard sensing capabilities, we simulate the robot’s perception using idealized LiDAR sensors in both 2D and 3D environments. Specifically, we assume the availability of a 2D planar LiDAR in two-dimensional settings, which provides radial distance measurements along uniformly distributed rays over $360^\circ$ within a fixed maximum range $R_s$. In 3D, we simulate a spherical LiDAR that returns point cloud data sampled uniformly, covering the robot’s surroundings up to a sensing radius $R_s$. In both cases, the sensor is characterized by its range $R_s$ and an angular resolution. 

For the 2D case, we consider the free space illustrated in Figure \ref{fig:2Dsimulation} where the obstacles boundaries are defined by interpolating a set of points with cubic splines. We take the desired goal position at $x_d = (4,-1)$, the robot's radius $R = 0.34$, the controller design parameter $\epsilon = 0.06$, $\mu=0.6$, $\nu = 1$ and the gain matrix $\mathbf{P} = [\begin{smallmatrix} .4 & .2 \\ .2 & .8 \end{smallmatrix}]$. The maximum range of the 2D sensor is $R_s=3$ and its resolution is $1^\circ$. In the 3D case, we consider similar parameters as in the 2D simulation example except for the gain matrix $\mathbf{P} = \Big[\begin{smallmatrix} 1 & 0 & 0 \\ 0 & 1 & .5\\ 0 &.5 &2 \end{smallmatrix}\Big]$, the 3D LiDAR resolution of $2^\circ$ and the goal $(4,7,1)$. Figure \ref{fig:3Dsimulation} illustrates the resulting trajectories for a 3D environment filled with convex and non-convex obstacles. 

\begin{figure}[h!]
\centering
    \begin{subfigure}[b]{0.48\linewidth}
        \centering
        \includegraphics[width=\linewidth]{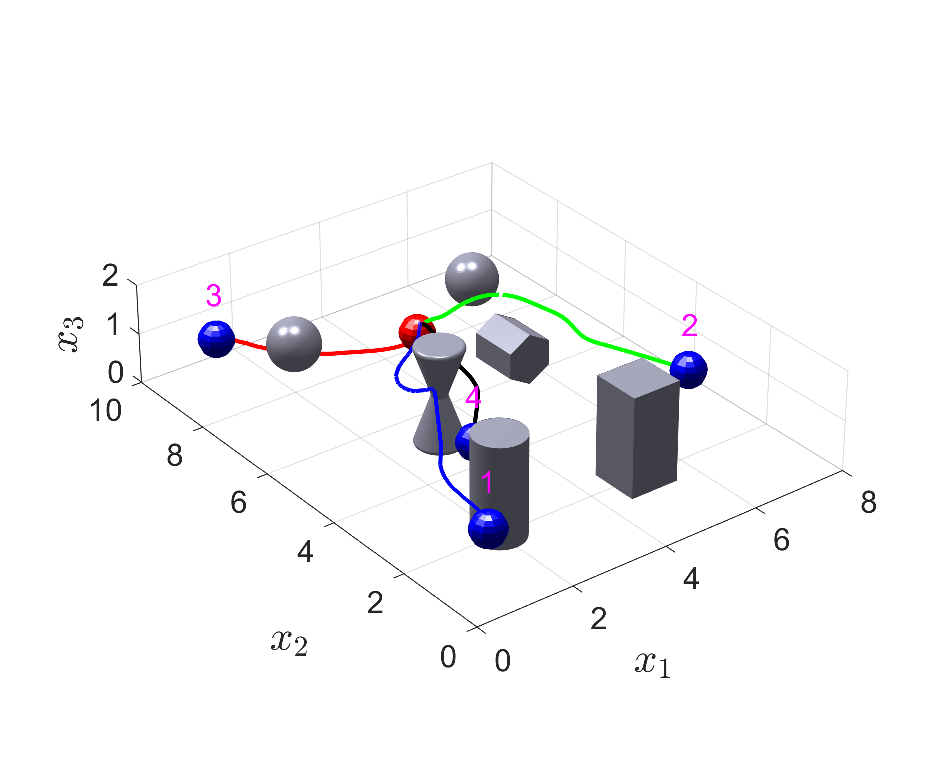}
        \caption{Perspective view.}
    \end{subfigure}
    \hfill
    \begin{subfigure}[b]{0.48\linewidth}
        \centering
        \includegraphics[width=\linewidth]{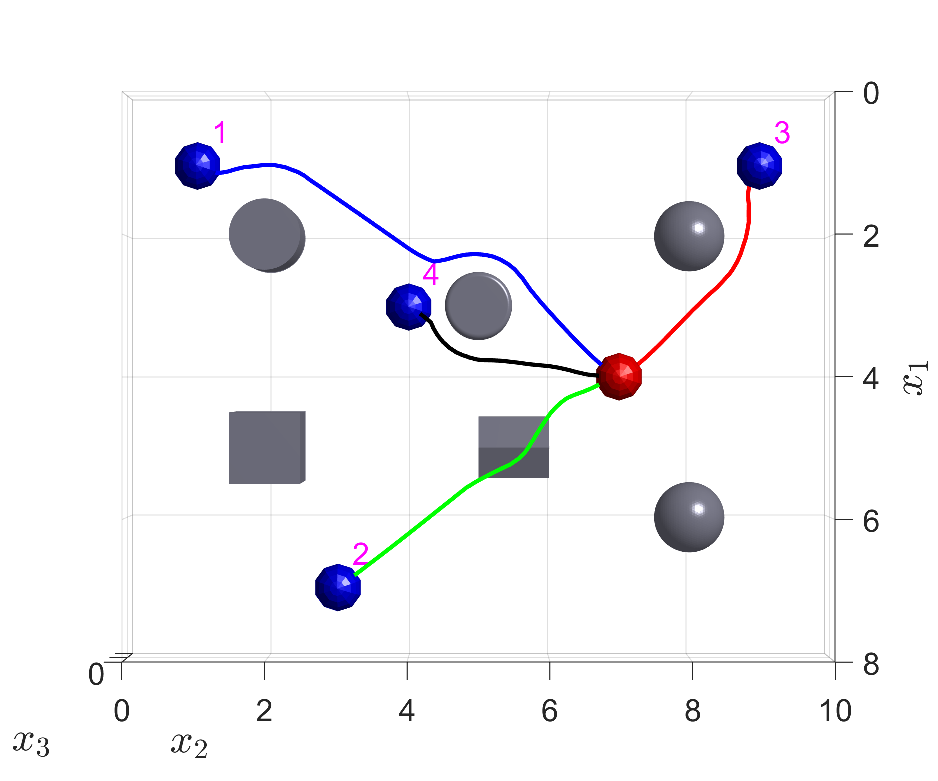}
        \caption{Top view.}
    \end{subfigure}

    \caption{Resulting trajectories (blue) of the robot under our approach in a 3D environment filled with convex and non-convex obstacles starting at a set of initial positions (blue) away from the goal (red) while avoiding the obstacles (gray). See video: \href{https://youtu.be/WZ40ftBPudk}{https://youtu.be/WZ40ftBPudk}}
    \label{fig:3Dsimulation}
\end{figure}
\begin{figure}[h!]
    \centering

    \begin{subfigure}[b]{0.5\linewidth}
        \centering
        \includegraphics[width=\linewidth]{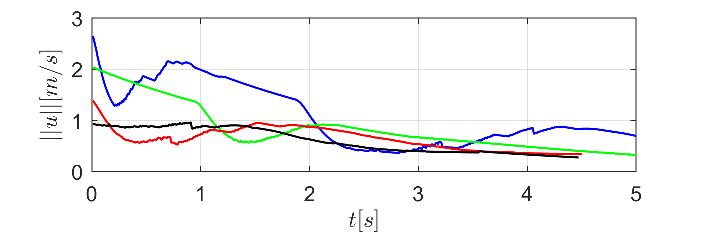}
        \caption{}
        \label{fig:u 3D}
    \end{subfigure}
    \vfill
    \begin{subfigure}[b]{0.5\linewidth}
        \centering
        \includegraphics[width=\linewidth]{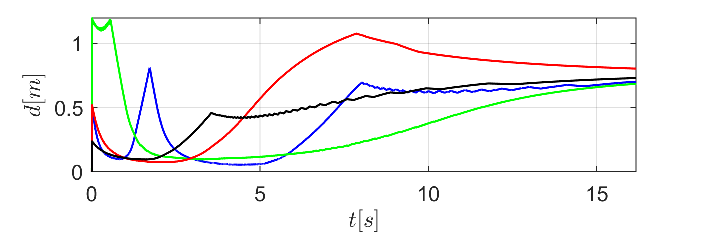}
        \caption{}
        \label{fig:d1 3D}
    \end{subfigure}

    \vfill
    \begin{subfigure}[b]{0.5\linewidth}
        \centering
        \includegraphics[width=\linewidth]{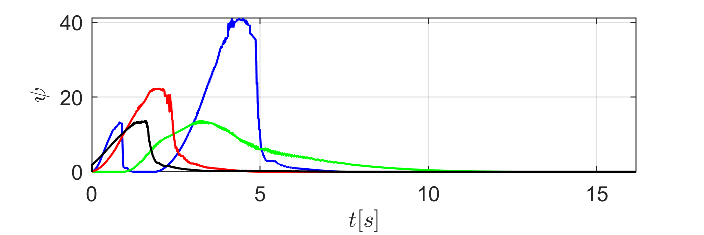}
        \caption{}
        \label{fig:psi1 3D}
    \end{subfigure}

    \caption{time evolution of (\subref{fig:u 3D}) the control input, (\subref{fig:d1 3D}) the distance to the obstacles and (\subref{fig:psi1 3D}) the penalty scaling function for the 3D case and the initial conditions: 1 in (blue), 2 in (green), 3 in (red) and 4 in (black).}
    \label{fig:u d psi 3D}
\end{figure}

Figures~\ref{fig:u d psi} and~\ref{fig:u d psi 3D} illustrates the temporal evolution of the key quantities involved in the closed-loop behavior. The minimum distance profiles consistently remain above the safety threshold, demonstrating that the robot successfully maintains a collision-free trajectory under all initial conditions. The behavior of the penalty scaling function further illustrates its activation mechanism where its value increases as the robot approaches the dilated boundary of an obstacle and when the nominal control direction aligns with the obstacle’s normal.

\section{Conclusion}\label{section: Conclusion}
This paper introduced SPF (Safe Penalty-based Feedback), a novel penalty-based framework for reactive obstacle avoidance that guarantees both safety and convergence using only local sensory information. By casting the controller as the closed-form solution of an unconstrained optimization problem, we obtained a smooth feedback law that minimally deviates from a given nominal controller while ensuring collision avoidance.

The theoretical results establish safety guarantees for arbitrary nominal controllers and ensure almost global asymptotic stability (AGAS) under a curvature condition when the nominal controller is a gradient descent of a potential function. Specifically, we require that, at any undesired equilibrium, there exists a tangent direction along which the curvature of the level set of the potential function is smaller than the curvature of the obstacle boundary. The approach was validated through 2D and 3D simulations, which demonstrate smooth, safe, and convergent behavior in cluttered environments, and highlight the simplicity of implementation.

Future directions include extending the method to handle multiple close obstacles as discussed in Remark~\ref{remark:multObs}, developing adaptive mechanisms to adjust the curvature of the level sets as discussed in \ref{remark:curvature}, and investigating both theoretical and practical benefits of such adaptations. Another promising extension involves adapting the framework to systems with higher-order dynamics and accounting for moving obstacles.


\bibliographystyle{unsrtnat}
\bibliography{references}  






\end{document}